\documentclass[12pt,letterpaper]{article}

\usepackage{amsmath}
\usepackage{amsfonts}
\usepackage{amssymb}
\usepackage{latexsym}
\usepackage{graphicx}
\usepackage{color}
\usepackage{url}
\usepackage{xspace} 
\usepackage{balance}
\usepackage{authblk}
\definecolor{ForestGreen}{rgb}{0.1333,0.5451,0.1333}
\usepackage[letterpaper,
            colorlinks,linkcolor=ForestGreen,citecolor=ForestGreen,
            backref,
            bookmarks,bookmarksopen,bookmarksnumbered]
           {hyperref}

\newtheorem{lemma}{Lemma}
\newtheorem{theorem}{Theorem} 
\newtheorem{proof}{Proof}

\newcommand{\One}[1]{\ensuremath{{\mathbf 1}\left(#1\right)}}
\newcommand{\fennel}{{\sc Fennel}\xspace}
\newcommand{\metis}{\textsc{Metis}\xspace}

\newcommand{\Prob}[1]{\ensuremath{{\bf{Pr}}\left[{#1}\right]}}
\newcommand{\Mean}[1]{\ensuremath{{\mathbb E}\left[{#1}\right]}}

\newcommand{\NPhard}{{\ensuremath{\mathbf{NP}}-hard}\xspace}

\newcommand{\whp}{\textit{whp}\xspace}
\newcommand{\Var}[1]{{\mathbb Var}\left[{#1}\right]}

\newcommand{\Egypt}{\textsc{EGyPT}\xspace}
\newcommand{\egypt}{\textsc{EGyPT}\xspace}

\newcommand{\Graphlab}{\textsc{GraphLab}\xspace}
\newcommand{\Pregel}{\textsc{Pregel}\xspace}

\newcommand{\hide}[1]{}

\newcommand{\field}[1]{\mathbb{#1}} 

\newcommand{\spara}[1]{\smallskip\noindent{\bf #1}}

\newcommand{\squishlist}{
 \begin{list}{$\bullet$}
  {  \setlength{\itemsep}{0pt}
     \setlength{\parsep}{3pt}
     \setlength{\topsep}{3pt}
     \setlength{\partopsep}{0pt}
     \setlength{\leftmargin}{2em}
     \setlength{\labelwidth}{1.5em}
     \setlength{\labelsep}{0.5em}
} }
\newcommand{\squishlisttight}{
 \begin{list}{$\bullet$}
  { \setlength{\itemsep}{0pt}
    \setlength{\parsep}{0pt}
    \setlength{\topsep}{0pt}
    \setlength{\partopsep}{0pt}
    \setlength{\leftmargin}{2em}
    \setlength{\labelwidth}{1.5em}
    \setlength{\labelsep}{0.5em}
} }

\newcommand{\squishdesc}{
 \begin{list}{}
  {  \setlength{\itemsep}{0pt}
     \setlength{\parsep}{3pt}
     \setlength{\topsep}{3pt}
     \setlength{\partopsep}{0pt}
     \setlength{\leftmargin}{1em}
     \setlength{\labelwidth}{1.5em}
     \setlength{\labelsep}{0.5em}
} }

\newcommand{\squishend}{
  \end{list}
}

\newtheorem{definition}{Definition} 

\begin{document}

\title{Streaming Graph Partitioning in the Planted Partition Model}

\author{
  Charalampos E. Tsourakakis\\
  Brown University\\
  \texttt{charalampos\_tsourakakis@brown.edu} 
}

\maketitle \sloppy

\begin{abstract}

The sheer increase in the size of graph data has created 
a lot of interest into developing efficient distributed graph processing frameworks. 
Popular existing frameworks such as \Graphlab and \Pregel 
rely on balanced graph partitioning in order to minimize communication  
and achieve work balance.   

In this work we contribute to the recent research line   
of streaming graph partitioning \cite{stantonstreaming,stanton,fennel}  
which computes an approximately balanced $k$-partitioning of the vertex set of a graph using 
a single pass over the graph stream  using degree-based criteria. 
This graph partitioning framework is well tailored to processing large-scale and dynamic graphs.
In this work we introduce the use of higher length walks for streaming graph partitioning
and show that their use incurs a minor computational cost which can significantly 
improve the quality of the graph partition. 
We perform an average case analysis of our algorithm using 
the planted partition model \cite{condon2001algorithms,mcsherry2001spectral}.
We complement the recent results of Stanton \cite{stantonstreaming} by showing that 
our proposed method recovers the true partition with high probability even when the 
gap of the model tends to zero as the size of the graph grows. Furthermore, among the wide 
number of choices for the length of the walks we show that the proposed length is optimal.  
Finally, we conduct experiments which verify the value of the proposed method.
\end{abstract}

\section{Introduction}
\label{sec:intro}
The size of graph data that are required to be processed nowadays is massive.
For instance, the Web graph amounts to at least one trillion of links~\cite{google}
and Facebook in 2012 reported more than 1 billion of users 
and 140 billion of friend connections.
Furthermore, graphs of significantly greater size emerge by post-processing various other  
data such as image and text datasets. 
This sheer increase in the size of graphs has created a lot of interest 
in developing distributed graph processing systems \cite{DBLP:conf/icdm/KangTF09,DBLP:conf/uai/LowGKBGH10,malewicz}, 
in which the graph is distributed across multiple machines.
A key problem towards enabling efficient graph computations in such  systems is the \NPhard problem of 
{\em balanced graph partitioning}.  High-quality partitions ensure low volume of communication and work balance. 

Recently, Stanton and Kliot \cite{stanton}  introduced a streaming graph partitioning model. 
This line of research despite being recent and lacking theoretical understanding 
has already attracted a lot of interest. 
Several existing systems have incorporated this model such as PowerGraph \cite{gonzalez12powergraph}.
The framework of \fennel has been adapted by PowerLyra \cite{powerlyra}, 
which has been included into the most recent \Graphlab version 
\cite{bickson,DBLP:conf/uai/LowGKBGH10} yielding significant speedups for various iterative computations. 
Stanton performed an average case analysis of two streaming algorithms,
and explained their efficiency despite the pessimistic worst case analysis \cite{stantonstreaming}. 
Despite the fact that existing established heuristics such as \metis 
typically outperform streaming algorithms, the latter are well tailored to today's needs 
for processing dynamic graphs and big graph data which do not fit in the main memory. 
They  are computationally cheap and provide partitions of comparable quality. 
For instance, \fennel  on the Twitter network with more than 1.4 billion edges performed 
comparably well with \metis for a wide variety of settings, requiring 40 minutes 
of running time, whereas \metis 8$\frac{1}{2}$ hours.  
 
So far, the work on streaming graph partitioning is based on computing the degrees 
of incoming vertices towards each of the $k$ available machines. Equivalently, 
this can be seen as performing a one step random walk from the incoming vertex. 
A natural idea which is used extensively in the literature of graph 
partitioning \cite{lovasz1990mixing,spielman2008local,zhou2004clustering}
is the use of higher length walks. In this work we introduce this idea 
in the setting of streaming graph partitioning. 
At the same time we maintain the time efficiency of streaming graph partitioning 
algorithms which make them attractive to various 
graph processing systems \cite{gonzalez12powergraph,DBLP:conf/icdm/KangTF09,DBLP:conf/uai/LowGKBGH10, malewicz}. 

{\bf Summary of our contributions}. 
Our contributions can be summarized as follows:
 
\squishlist

\item Our proposed algorithm introduces the idea of using higher 
length  walks for streaming graph partitioning. It   incurs a negligible computational cost
and significantly improves the quality of the partition. 
We perform an average case analysis  on the planted partition model, 
c.f. Section~\ref{subsec:models} for the description
of the model. We complement the recent results of \cite{stantonstreaming} which require that 
the gap $p-q$ of the planted partition model
is constant in order to recover the partition \whp\footnote{An event $\mathcal{E}_n$ 
occurs {\em with high probability}, or \whp\ for brevity, if
$\lim_{n\rightarrow\infty}\Prob{\mathcal{E}_n}=1$.}, by allowing 
gaps $p-q$ which asymptotically tend to 0 as $n$ grows. 

\item Among the wide number of choices for the length of $t$-walks where $t \geq 2$, 
we show that walks of length 2 are optimal as they allow the smallest 
possible gap for which we can guarantee recovery \whp, c.f. Section~\ref{subsec:pathtclassification}. 

\item We evaluate our method on the planted partition model and we provide 
a machine learning application which illustrates the benefits of our method.

\squishend

{\bf Structure of the paper}. The remainder of the paper is organized as follows. 
Section~\ref{sec:related} discusses briefly related work. 
Section~\ref{sec:algorithm} presents the proposed algorithm and its analysis.
Section~\ref{sec:experiments} presents the experimental evaluation. 
Section~\ref{sec:concl} concludes.

\section{Related Work} 
\label{sec:related}
Section~\ref{existing} briefly reviews  existing work. Section~\ref{prelim} 
states the probabilistic tools we use in Section~\ref{sec:algorithm}.
A complete account can be found in \cite{dubhashi2009concentration,mitzenmacher2005probability}.

\subsection{Existing Work} 
\label{existing}

\spara{ Balanced graph partitioning.}  The {\em balanced graph partitioning} problem is a classic \NPhard problem 
of fundamental importance to parallel and distributed computing. 
The input to this problem is an undirected graph $G(V,E)$ and an integer $k \in \field{Z}^+$, the 
output is a partition of the vertex set in $k$ balanced sets such that the number of edges across the clusters
is minimized. We refer to the $k$ sets as {\it clusters} or {\it machines} interchangeably.  
Formally, the balance constraint is defined by the imbalance parameter $\nu$. 
Specifically, the $(k,\nu)$-balanced graph partitioning 
asks to divide the vertices of a graph in $k$ clusters each of size 
at most $\nu \frac{n}{k}$, where $n$ is the number of vertices in $G$. 
The case $k=2,\nu=1$ is equivalent to the \NPhard   minimum bisection problem.
Several approximation algorithms, e.g., \cite{Feige:2002:PAM:586842.586910},
and heuristics, e.g., \cite{kl}, exist for this problem. 
When $\nu=1+\epsilon$ for any desired but fixed $\epsilon>0$ there exists a
$O(\epsilon^{-2} \log^{1.5}{n})$ approximation algorithm \cite{Andreev:2004:BGP:1007912.1007931}. 
When $\nu = 2$ there exists an $O(\sqrt{\log{k}\log{n}})$ approximation algorithm
based on semidefinite programming (SDP) \cite{krauthgamer}.  
Due to the practical importance of $k$-partitioning there exist several heuristics, 
among which \metis \cite{schloegel} stands out for its good performance. 
Survey \cite{christianschulz} summarizes many popular existing methods for the balanced
graph partitioning problem. 

\spara{ Streaming balanced graph partitioning.}  
Despite the large amount of work on the balanced graph partitioning problem, neither state-of-the-art
approximation algorithms nor heuristics  such as \metis are well tailored to the computational restrictions
that the size of today's graphs impose. 

Stanton and Kliot introduced the streaming balanced graph partitioning problem, 
where the graph arrives in the stream and decisions about the partition need to be taken with  on the fly quickly~\cite{stanton}. 
Specifically, the vertices of the graph arrive in a stream with the set of edges incident to them. 
When a vertex arrives, a partitioner decides  where to place the vertex ``on the fly'', using limited computational 
resources (time and space). 
A vertex is never relocated after it becomes assigned to one of the $k$ machines.  
\hide{ A realistic assumption that can be used in real-world streaming graph partitioners is the existence
of a small-sized buffer. Stanton and Kliot evaluate partitioners with or without buffers.}
In \cite{fennel} well-performing decision strategies for the partitioner were introduced
inspired by a generalization of optimal quasi-cliques \cite{tsourakakis2013denser} 
to $k$-partitioning.
Both \cite{stanton} and \cite{fennel} can be   adapted to edge streams.
Nishimura and Ugander \cite{nishimura2013restreaming} consider a realistic variation 
of the original problem where the constraint of a single pass over the graph stream 
is relaxed to allowing multiple passes. 
Stanton showed that streaming graph partitioning algorithms with a single pass 
under an adversarial stream order cannot approximate the optimal cut size within
$o(n)$. The same claim holds also for random stream orders \cite{stantonstreaming}.
Finally,  Stanton \cite{stantonstreaming} analyzes two variants of well performing algorithms 
from~\cite{stanton} on random graphs. Specifically, Stanton proves that if the graph 
$G$ is sampled according to  the planted partition model, then the two algorithms
despite their similarity can perform differently and that 
one of the two can recover the true partition \whp, assuming  that 
inter-, intra-cluster edge probabilities are constant, and their gap is a large constant.

\spara{ Planted partition model.} Jerrum and Sorkin \cite{jerrum1998metropolis} studied the planted bisection model, 
a random undirected graph with an even number 
of vertices. According to this model, each half of the vertices is assigned to one of two 
clusters. Then, the probability of an edge $(i,j)$ is $p$ if $i,j$ have the same 
color, otherwise $q<p$. We will refer to $p,q$ as the intra- and inter-cluster probabilities. 
Their difference $p-q$ will be referred to as the {\em gap} of the model. 
Condon and Karp \cite{condon2001algorithms} studied the generalization of the planted bisection problem where instead
of having only two clusters, there exist $k$ clusters of equal size.  
The probability of an edge is the same as in the planted bisection problem. 
They show that the hidden partition can be recovered \whp if the gap satisfies $p-q \geq n^{-1/2+\epsilon}$. 
McSherry \cite{mcsherry2001spectral} presents a spectral algorithm that recovers the hidden partition 
in a random graph \whp if $p-q = \Omega(n^{-1/2+\epsilon})$.  Recently, Van Vu gave an alternative 
algorithm \cite{DBLP:journals/corr/Vu14} to obtain McSherry's result. 
Zhou and Woodruff \cite{zhou2004clustering} showed that
if  $p=\Theta(1), q=\Theta(1), p-q=\Omega(n^{-1/4})$, then a simple algorithm based on squaring 
the adjacency matrix of the graph recovers  the hidden partition \whp.
 
\spara{Random walks.} The idea of using walks of length greater than one \cite{zhou2004clustering} 
is common in the general setting of graph partitioning. 
Lov\'{a}sz and Simonovits \cite{lovasz1990mixing} show
that random walks of length $O(\frac{1}{\phi})$ can be used 
to compute a cut with sparsity at most $\tilde{O}(\sqrt \phi)$ 
if the sparsest cut has conductance $\phi$. 
Later, Spielman and Teng \cite{spielman2008local} 
provided a local graph partitioning 
algorithm which implements efficiently the Lov\'{a}sz-Simonovits idea. 

\hide{
\spara{Clustering data streams.}
A general form of the clustering problem is the following: given an integer $k$ and a collection of 
$n$ points in a metric space, find $k$ cluster centers so that each point 
is assigned to one center and an objective which is a function of the sum 
of distances over all points from their assigned centers is minimized. 
The streaming setting has received a lot of attention due to the abundance 
of data in numerous data mining, pattern recognition and machine learning applications 
\cite{ailon2009streaming}.	
}

\subsection{Theoretical Preliminaries} 
\label{prelim}


A useful lemma that we use extensively is Boole's inequality, also known as the union bound. 

\begin{lemma}[Union bound] 
Let $A_1,\ldots,A_n$ be events in a probability space $\Omega$, then 
$$ \Prob{\cup_{i=1}^n A_i }  \leq \sum_{i=1}^n \Prob{A_i}.$$ 
\end{lemma}

The following theorem is due to Markov and its use is known as the first moment method. 

\begin{theorem}[First Moment Method] 
\label{markov}
Let $X$ be a non-negative random variable with finite expected value $\Mean{X}$. Then, for any 
real number $t > 0$, 

$$ \Prob{ X \geq t } \leq \frac{\Mean{X}}{t}.$$
\end{theorem}

The following theorem is due to Chebyshev and its use is known as the second moment method. 

\begin{theorem}[Second Moment Method]
\label{chebyshev} 
Let $X$ be a random variable with finite expected value $\Mean{X}$ and finite non-zero variance $\Var{X}$. Then, for any 
real number $t > 0$, 

$$ \Prob{ |X-\Mean{X}| \geq t } \leq \frac{\Var{X}}{t^2}.$$
\end{theorem}

\noindent Finally, we use the following Chernoff bounds for independent and negatively correlated
random variables. 

\begin{theorem}[Multiplicative Chernoff Bound]
\label{thrm:chernoff}
Let $ X=\sum_{i=1}^n X_i$ where  $ X_1,\ldots,X_n$ are independent random variables taking values in $ [0,1]$. 
Also, let $ \delta \in [0,1]$. Then, 

$$ Pr( X \leq (1-\delta)E(X) ) \leq e^{-\delta^2E(X)/2}.$$ 
\end{theorem}

\begin{definition}[Negatively correlated random variables]
Let $ X_1,\ldots,X_n$ be random binary variables. We say that they are negatively 
correlated if and only if for all sets $ I \subseteq [n]$ the following inequalities 
are true: $$ Pr(\forall i \in I: X_i=0) \leq \prod_{i \in I} Pr(X_i=0)$$ 
\noindent and  
$$ Pr(\forall i \in I: X_i=1) \leq \prod_{i \in I} Pr(X_i=1).$$ 
\end{definition}

\begin{theorem}
\label{thrm:chernoffneg}
Let $ X=\sum_{i=1}^n X_i$ where  $ X_1,\ldots,X_n$ are negatively correlated binary random variables.
Also, let $ \delta \in [0,1]$. Then, 

$$ Pr( X \leq (1-\delta)E(X) ) \leq e^{-\delta^2E(X)/2}.$$ 
\end{theorem}

\section{Proposed Algorithm}
\label{sec:algorithm}
Section~\ref{subsec:notation} introduces our notation
and Section~\ref{subsec:models} presents in detail the random graph model
we analyze. Section~\ref{subsec:simplelemmas} provides two 
useful lemmas used in Sections~\ref{subsec:path2classification}  and~\ref{subsec:pathtclassification}. 
Section~\ref{subsec:path2classification} shows an efficient way to recover the true partition of the planted
partition model \whp\ using walks of length $t=2$ even when $p-q=o(1)$, i.e., the gap 
is asymptotically equal to 0. 
Section~\ref{subsec:pathtclassification} shows that when the length of the walk $t$  
is set to 2, then we obtain the smallest possible {\em gap} $p-q$ for which we can 
guarantee recovery of the partition \whp. 
We do not try to optimize constants in our proofs, since we are interested in asymptotics.
Our proofs use the elementary inequalities $1-p \leq e^{-p}$, ${n \choose k} \leq \Big( \frac{en}{k}\Big)^k$
and the probabilistic tools presented in Section~\ref{prelim}.

\subsection{Notation}
\label{subsec:notation}

Let $G(V,E)$ be a simple undirected graph, where $|V|=n, |E|=m$. 
We define a \emph{vertex partition}  ${\mathcal P}=(S_1,\ldots,S_k)$ as a family 
of pairwise disjoint vertex sets whose union is $V$. 
Each set $S_i$ is assigned to one of $k$ machines, $i=1,\ldots,k$. 
We refer to each $S_i$ as a {\it cluster} or {\it machine}.
Throughout this work, we assume $k=\Theta(1)$. Let ${\partial e(\mathcal{P})}$ be the set of edges that cross partition boundaries, 
i.e., ${\partial e(\mathcal{P})} = \cup_{i=1}^k e(S_i,V\setminus S_i)$. 
The fraction of edges cut $\lambda$ is defined as  $ \lambda = \frac{|{\partial e(\mathcal{P})}|}{m}$.
The imbalance factor or  normalized maximum load $\rho$ is defined as 
$\rho  =  \max\limits_{1 \leq i \leq k}  \frac{|S_i|}{\frac{n}{k}}$. 
Notice that $\rho$ satisfies the double inequality $1 \leq \rho \leq k$.
When $\rho=1$ we have a perfectly balanced partition. At the other extreme, when $\rho=k$ 
all the vertices are placed in one cluster, leaving $k-1$ clusters empty. 
In practice, there is a constraint $\rho \leq \nu$ where $\nu$ 
is a value imposed by application-specific restrictions. Typically, $\nu$ is close to 1. 
\hide{ Let $B$ denote the size of the buffer. We will use the terms buffer size and seed size interchangeably. }
We omit floor and ceiling notation for simplicity, our results remain valid. 

\subsection{Planted Partition Model} 
\label{subsec:models} 

The model $G(n,\Psi,P)$  is a generalization of the classic Erd\"os-R\'enyi graph \cite{bollobas}. 
The first parameter $n$ is the number of vertices. 
The second parameter of the model is the function $\Psi:[n] \rightarrow [k]$ which maps each vertex
to one of $k$ clusters. 
We  refer to $C_i=\Psi^{-1}(i)$ as the $i$-th cluster, $i=1,\ldots,k$.  
The third parameter $P$ is a $k\times k$ matrix such that $0 \leq P_{ij} \leq 1$ for all $i,j=1,\ldots,k$
which specifies the probability distribution over the edge set. 
Specifically, a graph $G \sim G(n,\Psi,P)$ is generated by adding an edge with probability 
$P(\Psi(u),\Psi(v))$ between  each pair of vertices $\{u,v\}$. 
Notice, that when all the entries of $P$ are  equal to $p$, then 
$G(n,\Psi, P)$ is equivalent to the $G(n,p)$ model.
In this work we are interested into graphs that exhibit clustering. The planted partition model we 
analyze is the same as in \cite{condon2001algorithms,mcsherry2001spectral,stantonstreaming,zhou2004clustering}: 
$P_{ij} = p \One{i=j} +q \One{ (i\neq j)}$ for all $i,j=1,\ldots,k$. 
We assume that $p>q=\Theta(1)$. For simplicity, we refer to this version of $G(n,\Psi,P)$
as $G(n,k,p,q)$. 

\subsection{Results and useful lemmas}
\label{subsec:simplelemmas}

We prove two simple lemmas that we use in the analysis of our algorithms. 
We refer to the vertex that arrives exactly after $i-1$ vertices as the $i$-th vertex. 
The first lemma states that after a super-logarithmic number 
of vertices, we have seen a fraction of $\frac{1}{2k}$ 
vertices from each cluster of the planted partition \whp. 

\begin{lemma}
\label{lem:lem1}  
For all  $i \geq D \log^{1+\delta} (n)$ where $D,\delta>0$ are any positive constants,  
there exist at least $\frac{i}{2k}$ vertices that have already arrived from 
each cluster $j$ \whp, $j=1,\ldots,k$.
\end{lemma} 

\begin{proof} 
Fix any index $i  \geq D \log^{1+\delta}(n)$. Define for each $j\in [k]$. 
the bad event $\mathcal{A}_j$ that cluster $j$  has less than $\frac{i}{2k}$ vertices.
after the arrival of $i$ vertices. 
Given our assumption on the random order of the stream, the first $i$ vertices form 
a random $i$-subset of $[n]$.   
Let $Y_j$ be the number of vertices from cluster $j$ among the first $i$ vertices. 
We observe that the distribution of $Y_j$ is the hypergeometric distribution $H(n, n/k, i)$\footnote{Recall, 
the hypergeometric distribution $H(N,R,s)$ is a discrete probability distribution that 
describes the probability of choosing $r$ red balls in $s$ draws of balls 
without replacement from a finite population of size $N$ containing exactly $R$ red balls.}. 
Therefore, we obtain the following exact expression for the probability of the bad event $\mathcal{A}_j$:

$$ \Prob{ \mathcal{A}_j } = \sum_{r=0}^{i/2k} \Prob{Y_j=r} = \sum_{r=0}^{i/2k}   \frac{ {n/k \choose r} {n(k-1)/k \choose i-r } }{ {n \choose i} }.$$

Even if an asymptotic analysis using Stirling's formula is possible, a less tedious approach is possible. 
We express $Y_j$ as the sum of $i$ indicator variables $Y_{j,1},\ldots,Y_{j,i}$ where 
$Y_{j,l}=1$ if and only if the $l$-th vertex $v$ has $\Psi(v)=j$, for all $l=1,\ldots,i$.
Clearly, $\Mean{Y_j} = \frac{i}{k}$. Notice that even if the indicator random variables 
are not independent, they are negatively correlated. Therefore, Theorem~\ref{thrm:chernoffneg} applies, 
obtaining 

$$ \Prob{ Y_j \leq (1-\frac{1}{2}) \frac{i}{k} } \leq e^{-\frac{i}{8k}} \leq e^{-\frac{D\log^{1+\delta}(n)}{8k}} \ll  o(n^{-1}).$$

\noindent The proof is completed by taking a union bound over $k$ machines and $(n-D\log^{1+\delta}(n)$ vertices.
Specifically, 
let $\mathcal{E}$ be the event that there exists an index $i \geq D \log^{1+\delta} (n)$ such that 
 $\cup_{j=1}^k \mathcal{A}_j$ is true. Then, 

$$ \Prob{\mathcal{E}} \leq (n-D \log^{1+\delta} (n))k o(n^{-1}) =o(1).$$ 
\end{proof}

The second lemma states that as soon as  $\log^{6+\delta} (n)$ vertices have arrived 
in the stream for any $\delta>0$, then  the remaining vertices have at least $\log^6 (n)$
neighbors in the set of already arrived vertices.

\begin{lemma}
\label{lem:lem2}  
Let $\delta>0$  be any positive constant. After $\log^{6+\delta}(n)$ vertices have arrived, all remaining vertices in the stream have 
$\log^6{(n)}$ neighbors which reside in the $k$ machines. 
\end{lemma}

\begin{proof}

Let $i = g(n) \geq \log^{6+\delta}(n)$, 
be the number of vertices that  have arrived in the 
incidence stream. 
Let $\mathcal{E}$ be the event that there exists a vertex $v$ that has arrived after the first $i$ vertices 
with less than $\log^6 (n)$ neighbors in the set of arrived vertices. We obtain that $\Prob{\mathcal{E}}=o(1)$ using
the following upper bound. 


\begin{align*}
\Prob{ \mathcal{E}} &\leq n \sum_{j=0}^{\log^6 (n)} {g(n) \choose j} \big(\frac{p}{1-q}\big)^{i}(1-q)^{g(n)} \leq n\sum_{j=0}^{\log^6 (n)} \big(\frac{g(n)e}{j}\big)^j \big(\frac{p}{1-q}\big)^{j}(1-q)^{g(n)} \\
&\leq  Cn \log^6 (n) \Big(\frac{pg(n)e}{j(1-q)}\Big)^{\log^6 (n)} e^{-qg(n)} =o(1).
\end{align*} 
\end{proof}

\subsection{Path-$2$ classification}
\label{subsec:path2classification} 

\hide{Intuitively, if we wish to assign a vertex to its correct machine which contains 
the set of vertices from its cluster, one may perform 
walks of length greater or equal than 1 and see for which machine 
the probability of escaping is the smallest possible. 
This is exactly what the seeding procedure of \cite{stantonstreaming} does:
it allows walks of length 1 to correctly distinguish the true cluster
for each vertex. We will refer to this as path-$1$ classification. 
The analysis of \cite{stantonstreaming} shows that the gap $p-q$ has to be constant
in order to recover the true planted partition function $\Psi$. }

The main theoretical result of this Section is that we can use 
paths of length 2 to recover the partition $\Psi$ \whp even when $p-q=o(1)$. 
Our algorithm avoids making final decisions for the first $B$ vertices until
the end of the process and uses the fact that vertices from the same cluster 
have more common neighbors compared to a pair of vertices from different clusters.
\hide{
Here, we show how paths of length 2 can result in a correct partition of the planted partition model
with the use of a small-sized buffer. 
In the next section we argue why this length results in optimal results. 
The idea of the algorithm is simple: avoid making final decisions for the first $B$ vertices until
the end of the process
and use the fact that vertices from the same cluster have more common neighbors
compared to a pair of vertices from different clusters. In other words, a random walk of length 2 
is more likely to end up in the same cluster from where it started compared to a different one. }
Our algorithm is shown below.

\medskip
\label{alg:alg2}
\framebox{
\begin{minipage}[b]{0.8\linewidth}
\begin{enumerate} 
\item Place the first $B=\log^{6+\delta}(n)$ vertices in any of the $k$ machines, marked as non-classified.  
Here, $\delta>0$ is any positive constant.
\item Let $S$ be a random sample of  size $3k\log n$ vertices from the set of $B$ non-classified vertices. 
\item Let $R$ be a random sample of  $\log^{6} n$ vertices from the set of $B$ non-classified vertices. 
\item For the $j$-th vertex, $B+1 \leq j \leq n$, do the following: 
	\begin{itemize}
               \item For each $x \in S$ compute the number of common neighbors of $j,x$ in $R$. 
               \item Let $M= (p^2+(k-1)q^2) \log^{6} n$. 
               Assign $j$ to the same cluster with a vertex $x^* \in S$ which has at least $M-M^{2/3}$
              common neighbors with $j$.  
              Ties are always assigned to the vertex with the smallest id.  
              Remove non-classified tag from $x^*$. 
	\end{itemize}
\item  Perform the same procedure for the remaining, if any, non-classified vertices. 
\end{enumerate} 
\end{minipage}
\medskip\par}

We prove the correctness of the algorithm. 
The next lemma states that when we obtain the random sample $S$, 
there always exists at least one vertex from each cluster of the partition. 
This is critical since our algorithm assigns each incoming vertex $v$
to a representative vertex from cluster $\Psi(v)$. 
Among the various possible choices for a representative of cluster $c$, 
we choose the vertex $u$ with the minimum vertex id, namely $u=\arg\min \{u \in S: \Psi(u)=c \}$.

\begin{lemma} 
\label{lem:step1alg2} 
Let $S$ be a random sample of size $3k\log n$ vertices from a population 
of $j \geq \log^{6+\delta} (n)$ vertices. 
Then, with high probability there exists at least one representative vertex from 
each cluster of the planted  partition in  $S$
\end{lemma}

\begin{proof} 
First, notice that by Lemma~\ref{lem:lem1} there exist at least $j/2k \geq \frac{\log^{6+\delta} (n)}{2k}$ vertices
from each cluster $i=1,\ldots,k$. 
Let $\mathcal{E}_i$ be the event of failing to sample at least one vertex from cluster $i$ of the partition,
$i=1,\ldots,k$. We can upper bound the probability of the union $\cup_{i=1}^k \mathcal{E}_i$ of these
bad events as follows:  

\begin{align*} 
\Prob{ \cup_{i=1}^k \mathcal{E}_i} &\leq k \Big( 1-\frac{ \tfrac{j}{2k}}{j} \Big)^{3k\log n}  \\ 
 &\leq ke^{-3k\log n/2k} = o(n^{-1}). 
\end{align*} 

\end{proof}

The next theorem is our main theoretical result. It states that the algorithm recovers 
the true partition $\Psi$ \whp. 

\begin{lemma}
If $p=\Theta(1),q=\Theta(1), p-q = \omega\big(\frac{1}{\log (n)}\big)$, then 
all vertices are classified correctly \whp. 
The algorithm runs in sublinear time. 
\end{lemma}

\begin{proof}

{\bf Correctness:} Let $j$ be the index of the incoming vertex,  $x \in S$.
We condition on the event $\mathcal{E}$ that there exists at least one 
vertex from each cluster of the planted partition in the sample $S$. 
Define $Y_x$  as the number of triples $(j,u,x)$ where $u \in R$ and 
$x \in S$ such that $\Psi(j)=\Psi(x)$. 
Similarly, let $Z_x$  be the number of triples $(j,u,x)$, where 
$u \in R$ and $x \in S$ such that $\Psi(j) \neq \Psi(x)$. 
The expected values of $Y_x,Z_x$ are respectively

$$\Mean{Y_x} = \big(p^2+(k-1)q^2\big) \log^{6}n,$$ 

\noindent and 

$$\Mean{Z_x} = \big(2pq+(k-2)q^2\big) \log^{6}n.$$

\noindent A direct application of the multiplicative Chernoff bound, see Theorem~\ref{thrm:chernoff}, yields

\begin{align*}
\Prob{ Y_x \leq \Mean{Y_x} - \Mean{Y_x}^{2/3} } &\leq  e^{ - O(\Mean{Y_x}^{1/3})} = n^{-O(\log n)}\\  
\end{align*}
and 
\begin{align*}
\Prob{  Z_x \geq  \Mean{Z_x} + \Mean{Z_x}^{2/3} } &\leq  e^{- O(\Mean{Z_x}^{1/3})} =n^{-O(\log n)}\\  
\end{align*}

\noindent Furthermore, due to our assumption on $p,q=\Theta(1)$ and the gap $p-q \gg \tfrac{1}{\log n}$ we obtain 

\begin{align*}
(p-q)^2 \log^{6}n &\gg \Big( (p^2+(k-1)q^2)^{2/3} +(2pq+(k-2)q^2)^{2/3}\Big) \log^4 n \Leftrightarrow\\ 
\Mean{Y_x} - \Mean{Z_x} &\gg \Mean{Y_x}^{2/3}+\Mean{Z_x}^{2/3} \Leftrightarrow \Mean{Z_x} + \Mean{Z_x}^{2/3} \ll \Mean{Y_x} - \Mean{Y_x}^{2/3}.
\end{align*}

\noindent The above inequalities suggest that the number of 2-paths between $j$ and any vertex 
$x \in S$ such that $\Psi(x)=\Psi(j)$ is significantly larger compared to the respective count
between $j,x'$  such that $\Psi(x')\neq \Psi(j)$ \whp. 
Let $\mathcal{B}_j$ be the even that $j$ is misclassified. Combining the above inequalities
with Lemma~\ref{lem:step1alg2} results in 

$$ \Prob{ \mathcal{B}_j  } \leq  \Prob{ \mathcal{\bar{E}}} + \Prob{ \mathcal{B}_j | \mathcal{E}  }  \Prob{ \mathcal{E}  }= o(n^{-1}) .$$

\noindent By a union bound over $O(n)$ vertices, the proof is complete. 

{\bf Running time:} The algorithm  can be implemented in  $O(n\text{poly}\log(n))$ time in expectation.
Sampling $O(\text{poly}\log(n))$ samples can be implemented in expected 
$O(\text{poly}\log(n))$ time, c.f. \cite{knuth2007seminumerical}. 
Also, checking whether a neighbor of an incoming vertex resides in a 
given machine can be done in $O(1)$ time by using hash tables to store
the information within each machine. 
Notice that the number of edges in $G$ is $O(n^2)$ and therefore
the proposed algorithm is a sublinear time algorithm. 
\end{proof}

\subsection{Path-$t$ classification} 
\label{subsec:pathtclassification}

\noindent We conclude this section by discussing the effect of the length of the walk $t \geq 2$.  
Intuitively, $t$ should not be too large, otherwise the random walk will mix. 
We argue, that among all constant lengths $t\geq 2$, the choice $t=2$ 
allows the smallest possible gap for which we can find the true partition $\Psi$. 
It is worth outlining that $t=2$ is also in favor of the graph partitioning efficiency as well, 
since the smaller $t$ is, the less operations are required. 
Our results extend the results of Zhou and Woodruff \cite{zhou2004clustering}
to the streaming setting.  Our main theoretical result is the following theorem.

\begin{theorem} 
\label{thrm2} 
Let $t\geq 2, t=\Theta(1)$ be the length of a walk. 
If $p,q=\Theta(1)$ such that $p(1-p),q(1-q)=\Theta(1)$, then 
$t=2$ results in the largest possible gap $p-q$ for which we can decide whether 
$\Psi(u)=\Psi(v)$ or not \whp, where $u\neq v \in [n]$. 
\end{theorem}

To prove Theorem~\ref{thrm2}, we compute first the expected number of walks of length $t$ between any two vertices in $G(n,k,p,q)$. 
Because of the special structure of the graph, we are able to derive an exact 
formula\footnote{Despite the large amount of work on the planted partition, 
we were not able to find a closed formula but only bounds in the existing literature.}. 

\begin{lemma} 
\label{lem:closedformula} 
Let $G \sim G(nk,k,p,q)$ and $p_t=A_{uv}^t, q_t=A_{uw}^t$ be the two types of entries that 
appear in $A^t$ depending on whether  $\Psi(u) = \Psi(v)$ and $\Psi(u) \neq \Psi(w)$
respectively. 
Then

$$  p_t= (k-1)  \frac{ n^{t-1} (p-q)^{t} }{k} + \frac{ n^{t-1} (p+(k-1)q)^{t}}{k},$$

 \noindent and

$$   q_t = -\frac{ n^{t-1} (p-q)^{t} }{k} + \frac{ n^{t-1} (p+(k-1)q)^{t}}{k}.$$ 
 \end{lemma}

\begin{proof}
Let $A$ be the $(p,q)$-adjacency matrix defined as $A_{uv}=p$ if $\Psi(u)=\Psi(v)$ and $A_{uw} = q$ if $\Psi(u) \neq \Psi(w)$
for each $u \neq v \in [n]$. 
It is easy to check that for any $t \geq 1$ the block structure of the planted partition
is preserved.
This implies that for any $t$, matrix $A^t$ has the same block structure as $A$
and therefore there are two types of entries in each row. 
Let $p_t,q_t$ be these two types of entries in $A^t$. For $t=1$, let $p_1=p$ and $q_1=q$.
Then, by considering the multiplication of the $u$-th line of 
$A^t$ with the $v$-th column of $A$ we obtain 

$$ p_{t+1} = pn p_t + (k-1) n q q_t,$$ 

\noindent and similarly by considering the multiplication of the $u$-th line of $A^t$ with the $w$-th column of $A$
we obtain 

$$ q_{t+1} = qn p_t+ pn q_t + \frac{k-2}{k} n q q_t.$$

We can write the recurrence in a matrix form. 

 $$
\left[ \begin{array}{c} p_{t+1} \\ q_{t+1} \end{array} \right] = M \times \left[ \begin{array}{c} p_{t} \\ q_{t}  \end{array} \right]
 $$

where $M =\begin{pmatrix} pn & (k-1) n q\\ qn  & pn  + \frac{k-2}{k} n q \end{pmatrix} $.
By looking the eigendecomposition of $M$, despite the fact that it is not symmetric, we can 
diagonalize it as 

$$ M = USU^{-1},$$

\noindent where 
$$ S= \begin{pmatrix} n(p-q) &  0 \\ 0  & n \big(p+(k-1)q \big) \end{pmatrix} $$

\noindent and 

$$U = \begin{pmatrix} -(k-1) &  1 \\ 1  &   1 \end{pmatrix}.$$

\noindent  Given the fact   

$$ M^k= U \begin{pmatrix} \big( n(p-q)\big)^k &  0 \\ 0  & \Big( n \big(p+(k-1)q \big) \Big)^k \end{pmatrix} U^{-1}$$

\noindent and simple algebraic manipulations (omitted)  we obtain that 

$$ p_t = (k-1)  \frac{ n^{t-1} (p-q)^{t} }{k} + \frac{ n^{t-1} (p+(k-1)q)^{t}}{k},$$ 

and 

$$ q_t  = -\frac{ n^{t-1} (p-q)^{t} }{k} + \frac{ n^{t-1} (p+(k-1)q)^{t}}{k}.$$ 
\end{proof}

\noindent Now, we are able to prove Theorem~\ref{thrm2}.

\begin{proof}[Theorem~\ref{thrm2}]
Let  $p_t = A_{uv}^t, q_t = A_{uw}^t$ where $u,v,w \in V(G)$ such that $\Psi(u)=\Psi(v) \neq \Psi(w)$
and $A$ is defined as in Lemma~\ref{lem:closedformula}. 
Also, define $\bar{A}$ to be the result of the randomized rounding of $A$.
By Lemma~\ref{lem:closedformula} we obtain 

$$p_t-q_t= \frac{n^{t-1} (p-q)^{t}}{k^{t-1}}.$$ 

\noindent Notice that we substituted $n$ by $n/k$ as we $G$ has $n$ vertices, 
with exactly $n/k$ vertices per cluster. 
Now, suppose $|A ^{t}_{uv} -  \bar{A}^{t}_{uv} |  \leq \gamma$,
where $\gamma>0$ is large enough such that the inequality holds \whp
and will be decided in the following. 
Then, if $\Psi(j_1) = \Psi( j_2) = \Psi(u)$ we obtain the following upper bound 

\begin{align*} 
| \bar{A}^{t}_{uj_1} -  \bar{A}^{t}_{uj_2} | &\leq |  A^{t}_{uj_1} - \bar{A}^{t}_{uj_1}| + | A^{t}_{uj_2} -  \bar{A}^{t}_{uj_2} | = 2\gamma.
\end{align*}

On the other hand if $\Psi(u)= \Psi( j_1) \neq  \Psi( j_2)$,
given that $|x| = |x+y-y| \leq |y|+|x-y| \rightarrow |x-y| \geq |x|-|y|$,
we obtain the following lower bound

\begin{align*} 
| \bar{A}^{t}_{uj_1} -  \bar{A}^{t}_{uj_2} | &\geq |  {A}^{t}_{uj_1} -   {A}^{t}_{uj_2} |  -  2\gamma = \frac{n^{t-1} (p-q)^{t}}{k^t} - 2\gamma.
\end{align*}

Therefore, if $\frac{n^{t-1} (p-q)^{t}}{k^{t-1}}-2\gamma  > 2\gamma$,  then 
there exists a signal that allows us to classify the vertex correctly. 
In order to find $\gamma$ we need to upper-bound the expectation of the non-negative random variable 
$Z=( \bar{A}^{t}_{uv} - A^{t}_{uv} )^2$. 

Zhou and Woodruff prove $\Mean{Z} = \Theta(n^{2t-3})$,c.f. Lemma 5, \cite{zhou2004clustering}. 
The proof of this claim is based on algebraic manipulation. 
It is easy to verify that $\Mean{Z}$ is dominated by the terms that correspond to two paths 
of length $t$ which overlap on a single edge. 
Applying Markov's inequality, see Theorem~\ref{markov}, we obtain 

$$ \Prob{ Z \geq  \Mean{Z} \log n } \leq \frac{1}{\log n}=o(1),$$ 

\noindent This suggests  setting $\gamma=n^{t-3/2}\sqrt{\log n}$ since then
$|A ^{t}_{uv} -  \bar{A}^{t}_{uv} |  \leq n^{t-3/2}\sqrt{\log n}$ \whp.
The gap requirement is 

$$\frac{n^{t-1} (p-q)^{t}}{k^{t-1}}  > 4n^{t-3/2}\sqrt{\log n} \rightarrow p-q = \Omega \Big( \big( \sqrt { \frac{\log n}{n}} \big)^{1/t}\Big).$$ 

\noindent This proves our claim, as for $t=2$ we obtain the best possible gap. 

\end{proof}

\section{Experimental Results}
\label{sec:experiments}
In this Section we present our experimental findings. We refer to our method
as \egypt ({\em E}fficient {\em G}raph {\em P}ar{\em T}itioning). 
Specifically, Section~\ref{subsec:setup}  describes the experimental setup. 
Section~\ref{results} presents the simulation results which verify the value of our proposed method.
Finally, Section~\ref{application} provides an application of our method on real-data. 
 
\subsection{Datasets}
\label{subsec:setup}

The experiments were performed on a single machine, with Intel Xeon CPU
at 2.83 GHz, 6144KB cache size and and 50GB of main memory. 
We have implemented \egypt in both \textsc{Java JDK 1.6}  and \textsc{Matlab} R2011a. 
The method we use to compare against for the 1-path classification is LWD, c.f. \cite{stanton},
the single streaming graph partitioning method for which we have theoretical insights \cite{stantonstreaming}. 
The simulation results in Section~\ref{results} were obtained using the \textsc{Java} code. 
The results in Section~\ref{application} were obtained from the \textsc{Matlab} implementation.

{\em Synthetic data:} We generate random graphs according to the planted partition model. 
We fix $q=0.05$ and we range the gap from $0.05$ until $0.95$ with a step of $0.05$. This 
results in pairs of $(p,q)$ values with $p$ ranging from $0.1$ until $1$. 
and $k$ ranging from $2$ to $16$ as successive powers of $2$ are qualitatively identical. 
Notice that given that $p,q=\Theta(1)$ despite the small number of vertices
the number of edges is large ranging from $2$ to $9.2$ million edges.  

Parameter $B$ in the interval $\{50,100,200,\ldots,1000\}$. 
The imbalance tolerance $\rho$ was set to 1, demanding equally sized clusters (modulo the remainder of n divided by k).
Since the function $\Psi$ is known, i.e.,  groundtruth is available, 
we measure the precision of the algorithm, as the percentage of the ${n \choose 2}$ 
relationships that it guesses correctly. 
The success of the algorithm is also judged in terms of the fraction of edges cut $\lambda$. 

{\em Real data:} An interesting question is what kind of real-world graphs does the planted
partiton model capture? Real-world networks are not modelled by the planted partition model \cite{amazing}. 
However, nearest neighbor graphs of well-clustered data points appear to be modelled 
by the planted partition model we analyzed in Section~\ref{sec:algorithm}.
As a cloud of points, we use a perfectly balanced set of 50\,000 digits (5\,000 digits for each digit 0,1,\ldots,9) from the MNIST  
database \cite{yann}. Each digit is a 28$\times$28 matrix which is 
converted in a 1-dimensional vector with 784 coordinates. 

\begin{figure}[!ht]
\centering
\begin{tabular}{@{}c@{}@{\ }c@{}}
\includegraphics[width=0.505\textwidth]{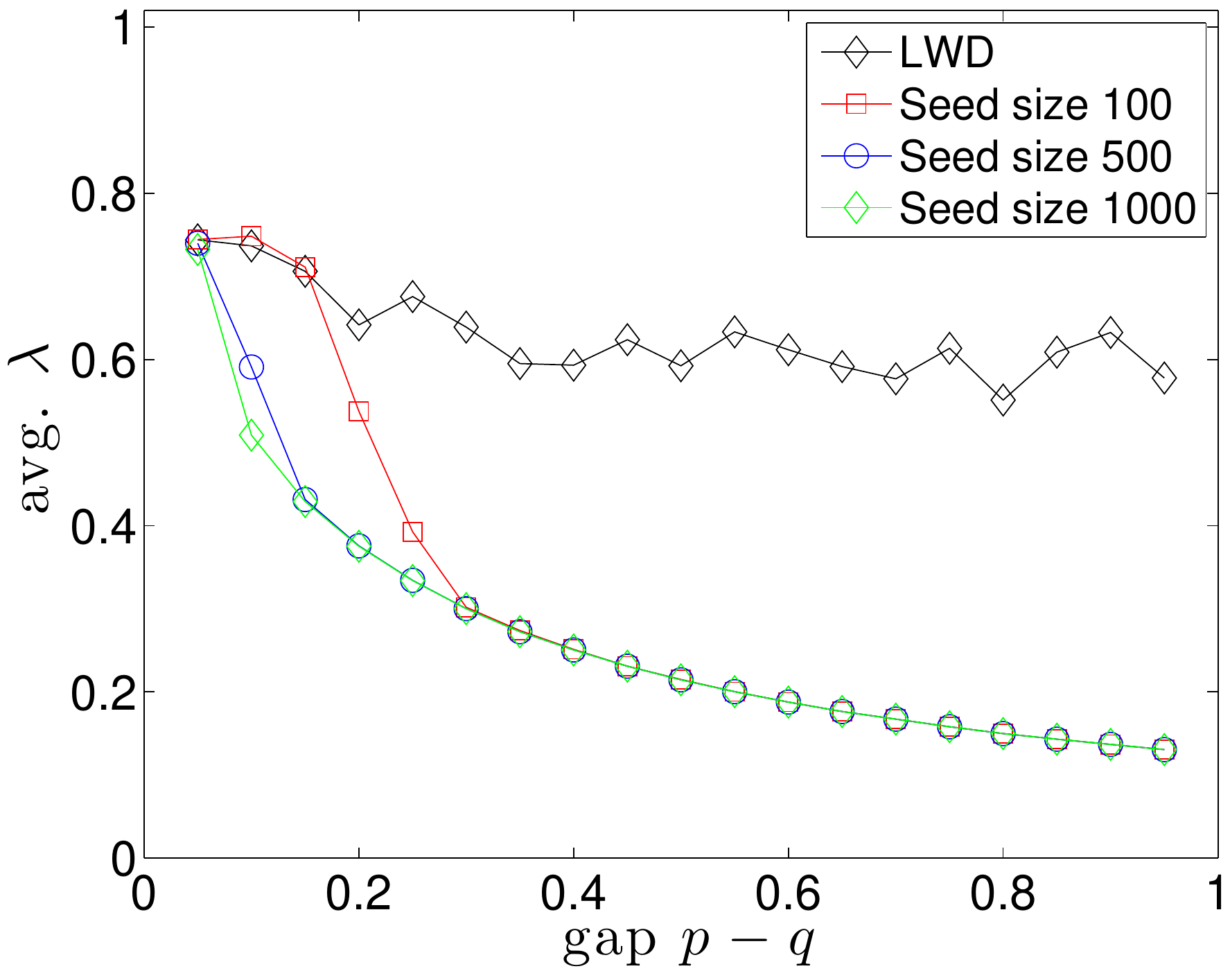} & \includegraphics[width=0.475\textwidth]{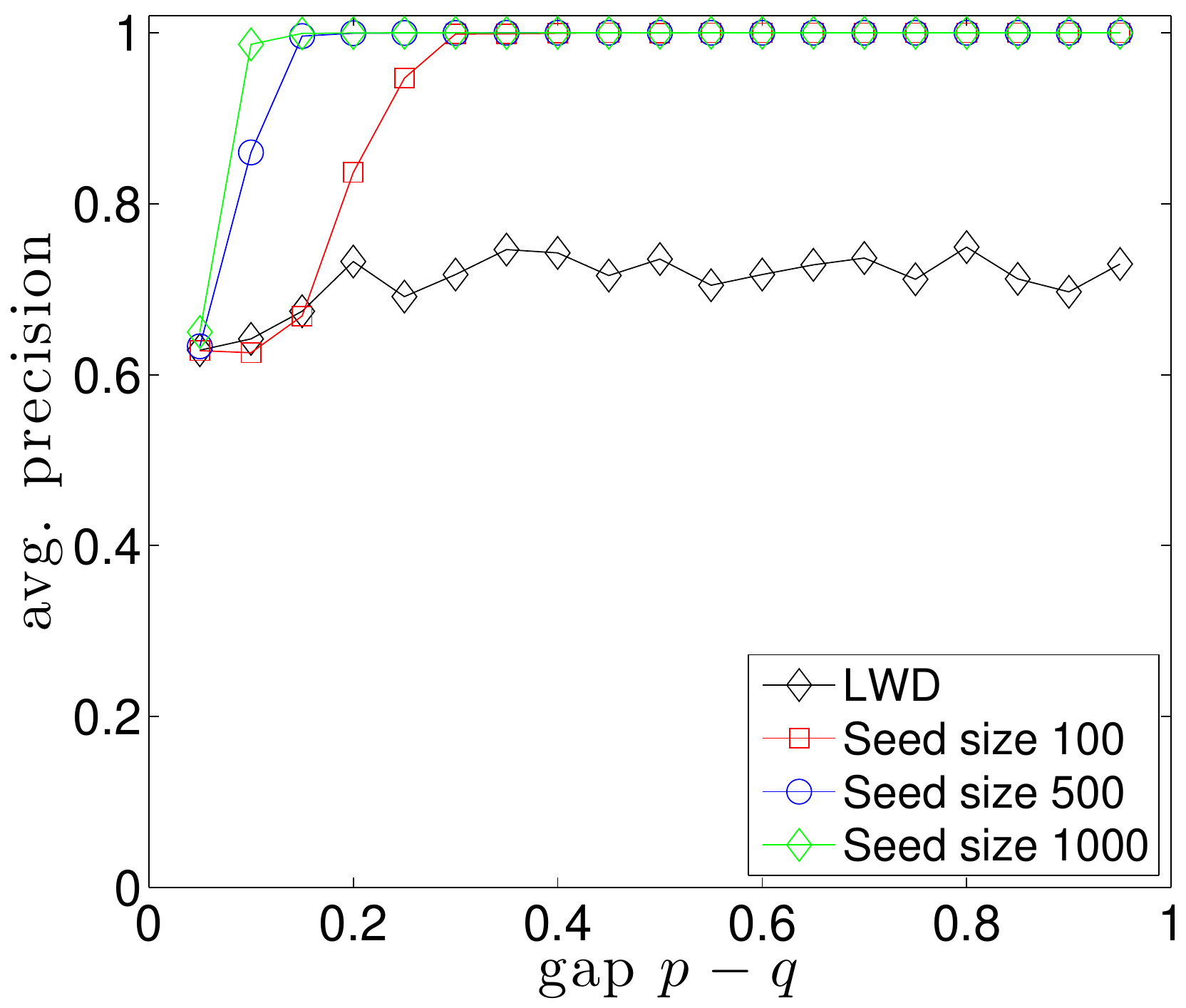} \\
(a) & (b) \\
\includegraphics[width=0.505\textwidth]{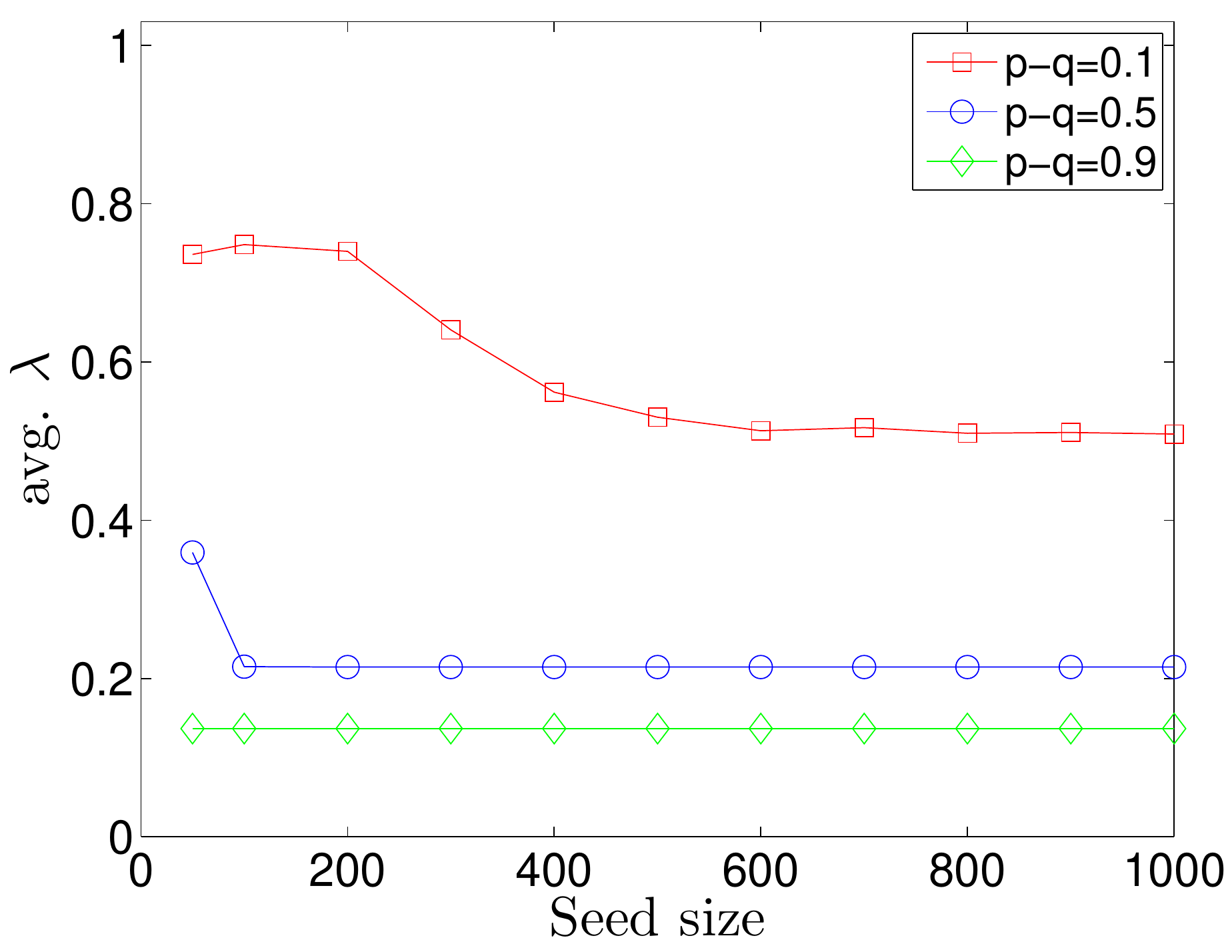} & \includegraphics[width=0.475\textwidth]{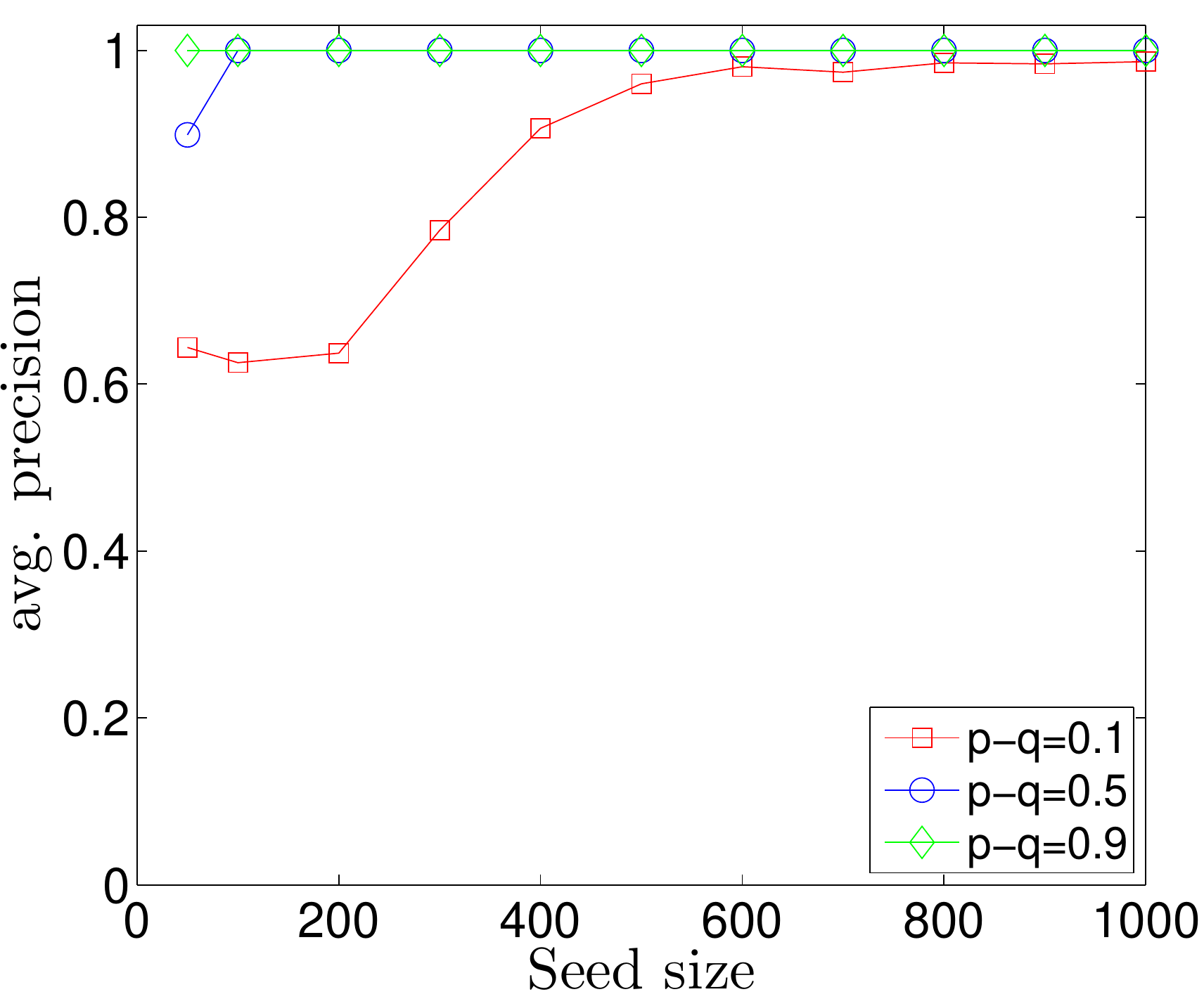}  \\
(c) & (d) 
\end{tabular}
\caption{\label{fig:fig1} (a) Average fraction
of edges cut $\lambda$ and (b)  average precision versus gap $p-q$ for LWD and \Egypt  for three different $B$ values (seed sizes).
(c) Average precision and (d) average $\lambda$, versus $B$ for three different 
gaps. All data points are averages over five experiments. Observed values are strongly concentrated around their corresponding averages.  }
\end{figure}

\subsection{Simulations}
\label{results}

Figures~\ref{fig:fig1}(a) and (b) plot   the average fraction of edges 
cut $\lambda$ and  the average precision   versus the gap $p-q$ for LWD and three runs 
of \egypt with different seed sizes, $B \in \{100,500,1000\}$. 
We observe that even for a small value of $B$, the improvement over LWD is significant. 
Furthermore, we observe that even for $p=1, q=0.05$ which corresponds to 
the 0.95 gap, LWD is not able to output a good quality partition. 
This shows that the analysis of \cite{stantonstreaming} needs even larger values 
of $n$ to recover the partition \whp. 
Furthermore, as  $B$ increases from 100 to 1000 
the quality of the final partition improves. 
It is worth emphasizing that the results we obtain are strongly concentrated around 
their corresponding averages.  The ratio of the variance over the mean squared was at most 0.0129 and typically
of the order $10^{-3}$ indicating a strong concentration according to Chebyshev's 
inequality, see Theorem~\ref{chebyshev}.
 
Finally, Figures~\ref{fig:fig1}(c) and (d)  plot the average fraction of edges cut $\lambda$ 
and the average precision versus the seed size $B$ for three different gaps, averaged over five 
experiments. Again, data points are concentrated around averages.  As expected, the smaller
the gap $p-q$ the larger the parameter $B$ has to be in order to obtain a given amount of precision. 
When the gap is large even for a small seed size $B=50$, \egypt obtains the correct partition. Notice that LWD 
cannot achieve this level of accuracy. At the same time the computational overhead is negligible, at the order of seconds.

\begin{figure}[!ht]
\centering
\begin{tabular}{c}
\includegraphics[width=0.65\textwidth]{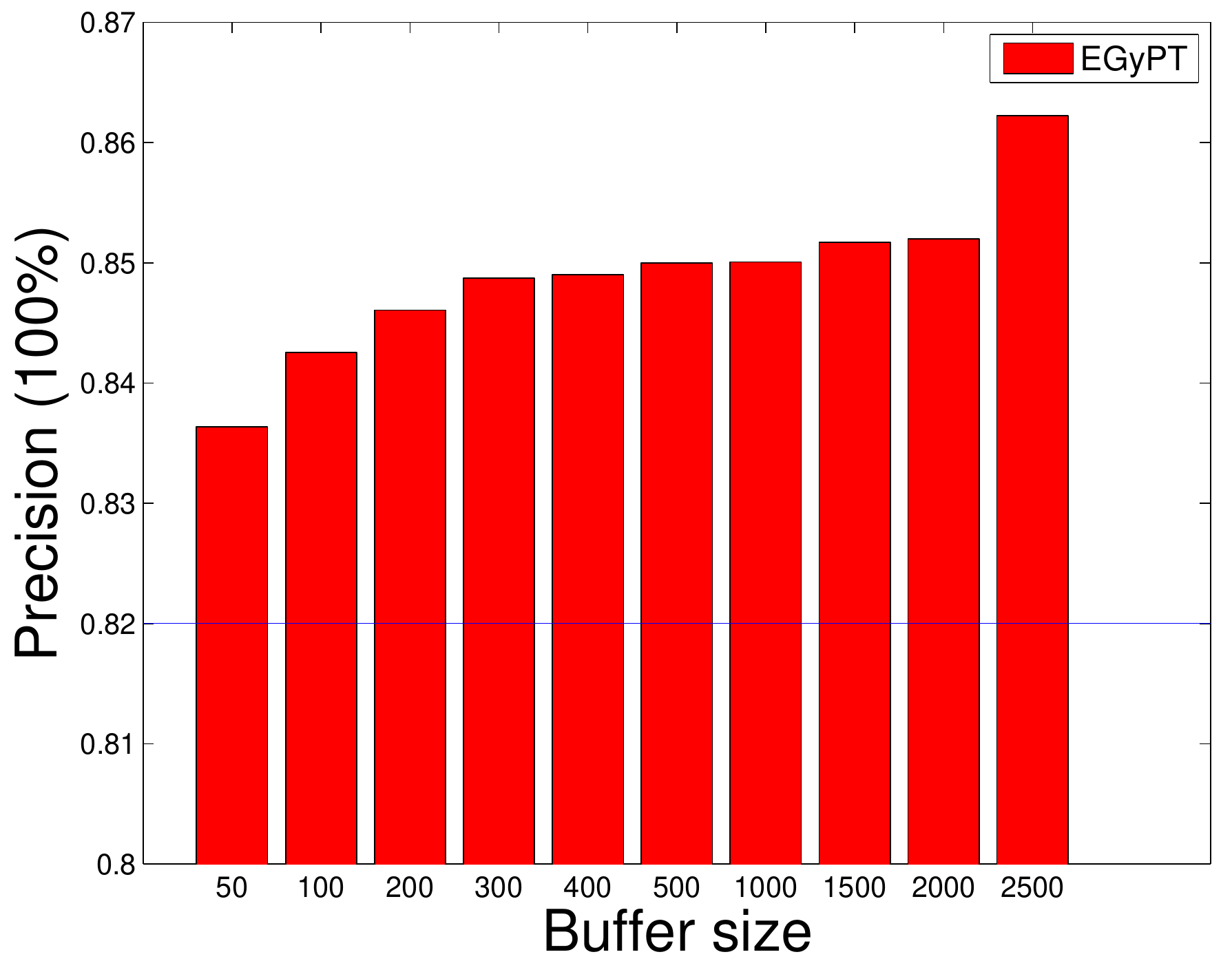}\\ 
\end{tabular}
\caption{\label{fig:fig6} Digits classification  }
\end{figure}

\subsection{A Machine Learning Application}
\label{application}

We consider a stream of data points where each data point $x \in \field{R}^{784}$ represents a digit from 0 to 9.
Whenever a new data point $x$ arrives, we find its $k'=5$ nearest neighbors among the $B$ first points. 
This is a variation of the well known $k'$-nearest neighbor graph \cite{maier2009optimal}. 
It is worth outlining that the planted partition model captures somes aspects of the resulting $k'$-nearest
neighbor graph. Specifically, the key property we are interested in is the 
community structure of the graph. 
For a given set of vertices $U \subseteq V$ the conductance of $S$ is defined $\phi(U)$ is ${ E(U,\bar{U})}/{ \text{vol}(U) } $ where 
$E(U,\bar{U}) =| \{ (u,v): u \in U, v \in \bar{U}, w(u,v) = 1\}|$ and 
$vol(U)=\sum_{v\in U} deg(v)$. The conductances of the 10 subsets of vertices  for each possible digit
are  0.0088, 0.0160, 0.0203, 0.0282, 0.0214, 0.0253, 0.0122, 0.0229, 0.0291, 0.0328 respectively.
Figure~\ref{fig:fig6} shows the improvement in the precision of the clustering as parameter $B$
increases from 0.1\% (50) to 5\% ($2\,500$) of the total number of data points ($50\,000$). The blue straight line shows the performance
of LWD. We observe the improvement over LWD even for $B=50(=0.1\% \times 50\,000)$
and the monotone increasing behavior of the precision as a function of $B$.

\section{Conclusions}
\label{sec:concl}
In this work we intoduce the natural idea of higher length walks 
for streaming graph partitioning, a recent line of research \cite{stanton}
that has already had a significant impact on various graph processing systems. 
We analyze our proposed algorithm \egypt in the planted partition model.

In future work, we plan to perform an average case  performance analysis of \egypt
using a random graph model with power law degree distribution and small separators, 
e.g., \cite{flaxman2006geometric,frieze2012certain}, and to explore the performance of \egypt on social networks.

\section{Acknowledgements} 
The author would like to thank Alan Frieze and Moez Draief for their feedback on the manuscript. 

\bibliographystyle{alpha}
\bibliography{ref}

\end{document}